\documentclass[11pt,a4paper]{article}

\usepackage{graphicx}
\usepackage{amsfonts}
\usepackage{amsmath}
\usepackage{amsthm}
\usepackage[mathcal]{euscript}      
\usepackage{bm}                     
\usepackage{color}
\usepackage{epstopdf}
\usepackage{amssymb}
\usepackage{mathrsfs,mathdots}
\usepackage{lscape}

\usepackage[paper=a4paper,dvips,top=3cm,left=2.4cm,right=2.4cm,
    foot=1.4cm,bottom=3cm]{geometry}

\begin{document}

\title{Octupolar Tensors for Liquid Crystals}
\author{Yannan Chen\footnote{%
    School of Mathematics and Statistics, Zhengzhou University, Zhengzhou 450001, China ({\tt ynchen@zzu.edu.cn}).
    This author was supported by the National Natural Science Foundation of China (Grant No. 11401539),
    the Development Foundation for Excellent Youth Scholars of Zhengzhou University (Grant No. 1421315070),
    and the Hong Kong Polytechnic University Postdoctoral Fellowship.}
  \and Liqun Qi\footnote{%
    Department of Applied Mathematics, The Hong Kong Polytechnic University,
    Hung Hom, Kowloon, Hong Kong ({\tt maqilq@polyu.edu.hk}).
    This author's work was partially supported by the Hong Kong Research Grant Council
    (Grant No. PolyU  501913, 15302114, 15300715 and 15301716).}
  \and Epifanio G. Virga\footnote{Mathematical Institute, University of Oxford, UK ({\tt virga@maths.ox.ac.uk}), \emph{on leave from} Dipartimento di Matematica, Universit\`{a} di Pavia,
    via Ferrata 5, I-27100 Pavia, Italy ({\tt eg.virga@unipv.it}). The work of this author was done while he was visiting the Oxford Centre for Nonlinear PDE at the University of Oxford, whose kind hospitality he gratefully acknowledges.}
}
\date{\today}
\maketitle

\begin{abstract}
  A third-order three-dimensional symmetric traceless tensor, called the \emph{octupolar} tensor, has been introduced to study tetrahedratic nematic phases in liquid crystals.
  The octupolar \emph{potential}, a scalar-valued function generated on the unit sphere by that tensor,
  should ideally have four maxima capturing the most probable molecular orientations (on the vertices of a tetrahedron), but it was recently found to possess an equally generic variant with \emph{three} maxima instead of four.
  It was  also shown that the irreducible admissible region for the octupolar tensor in a three-dimensional parameter space is bounded by a dome-shaped surface, beneath which is
  a \emph{separatrix} surface connecting the two generic octupolar states.
  The latter surface, which was obtained through numerical continuation,  may be physically interpreted as marking a possible \emph{intra-octupolar} transition.
  In this paper, by using the resultant theory of algebraic geometry and
  the E-characteristic polynomial of spectral theory of tensors,
  we give a closed-form, algebraic expression for both the dome-shaped surface and the separatrix surface.
  This  turns the envisaged intra-octupolar transition into a quantitative, possibly observable prediction.
  Some other properties of octupolar tensors are also studied.

\end{abstract}

\newtheorem{Theorem}{Theorem}[section]
\newtheorem{Definition}[Theorem]{Definition}
\newtheorem{Lemma}[Theorem]{Lemma}
\newtheorem{Corollary}[Theorem]{Corollary}
\newtheorem{Proposition}[Theorem]{Proposition}

\renewcommand{\hat}[1]{\widehat{#1}}
\renewcommand{\tilde}[1]{\widetilde{#1}}
\renewcommand{\bar}[1]{\overline{#1}}
\newcommand{\REAL}{\mathbb{R}}
\newcommand{\COMPLEX}{\mathbb{C}}
\newcommand{\SPHERE}{\mathbb{S}^2}
\newcommand{\diff}{\,\mathrm{d}}
\newcommand{\st}{\mathrm{s.t.}}
\newcommand{\T}{\top}
\newcommand{\vt}[1]{{\bf #1}}
\newcommand{\x}{\vt{x}}
\newcommand{\y}{\vt{y}}
\newcommand{\z}{\vt{z}}
\newcommand{\Ten}{\mathcal{T}}
\newcommand{\A}{\mathcal{A}}
\newcommand{\RESULTANT}{\mathrm{Res}}

\newpage
\section{Introduction}

Liquid crystals are well-known for visualization applications in flat panel electronic displays.
But beyond that, various optical and electronic devices, such as laser printers,
light-emitting diodes, field-effect transistors, and holographic data storage,
were invented with the development of bent-core (banana-shaped) liquid crystals \cite{Jak-13,LBS-07}. A third-order three dimensional symmetric traceless tensor was introduced in \cite{Fel-95}
to characterize condensed phases exhibited by bent-core molecules \cite{NSW-96,LNS-97}.
Based on such a tensor, the orientationally ordered octupolar (tetrahedratic) phase
has been both predicted theoretically \cite{LuR-02,BPC-02} and confirmed experimentally \cite{WNS-08}.
After that, the octupolar order parameters of liquid crystals have been widely studied \cite{BrP-10,Tur-11,GaV-16}. Generalized liquid crystal phases are also considered in \cite{Nissinen-16,Liu-16}, which feature octupolar order tensors among so many others.

Ideally, the octupolar order in a molecular ensemble of generalized liquid crystals is expected to exhibit four directions in space pointing towards the vertices of a tetrahedron, along which specific molecular axes are most likely to be oriented.
Fel \cite{Fel-95} first proposed to use a third-order three dimensional symmetric tensor $\A$,
which here is called the \emph{octupolar tensor}, to describe the tetrahedratic symmetry of the octupolar order.
According to the Buckingham's formula \cite{Buc-67,Tur-11} for the expansion in Cartesian tensors of the orientational probability density function for uniaxial nematics,
the octupolar tensor $\A$ is traceless, i.e., the trace of each slice matrix of $\A$ vanishes. As shown in \cite{GaV-16},
by a judicious choice of the Cartesian coordinate system,
such a tensor can be represented by three independent parameters, which there were called $\alpha_0,\beta_3$ and $\alpha_2$.



Virga \cite{Vi-15} and Gaeta and Virga \cite{GaV-16}, in their studies of
third-order octupolar tensors in two and three space dimensions, respectively, also introduced
the \emph{octupolar potential}, a scalar-valued function on the unit sphere obtained from the octupolar tensor.
In particular, Gaeta and Virga \cite{GaV-16} showed that the irreducible admissible region for
the octupolar potential is bounded by a surface in the three-dimensional parameter space which has the form of a \emph{dome} and, more importantly, that  there are indeed \emph{two} generic octupolar states,
divided by a \emph{separatrix} surface in paramter space. Physically,
the latter surface was interpreted as representing a possible intra-octupolar transition.

In this paper, by using the resultant theory of algebraic geometry
and the E-characteristic polynomial of the spectral theory of tensors,
we give a closed-form, algebraic expression for both the dome and the separatrix.
This turns the intra-octupolar transition envisioned in \cite{GaV-16} into a quantitative, possibly observable prediction.
Some other properties of octupolar tensors are also studied.

In Section~\ref{sec:2}, we prove that the traceless property of octupolar tensors
is preserved under orthogonal transformations. This property motives us to choose
a proper Cartesian coordinate system to reduce the independent elements of
the octupolar tensor $\A$ from seven to three.
By assuming that the North pole $(0,0,1)^\T$ is a maximum point of
the octupolar potential on the unit sphere,
we identify an irreducible, admissible region for only three independent parameters  of $\A$, which we shall also call
$\alpha_0,\beta_3 $ and $\alpha_2$ to ease comparison with \cite{GaV-16}.

Using the multi-polynomial resultant from algebraic geometry \cite{CLO-04},
we derive the resultant $\RESULTANT(\A\x^2)$ in Section~\ref{sec:3}. $\RESULTANT(\A\x^2)=0$
implies that $\A$ has zero E-eigenvalues. After that, we construct the E-characteristic polynomial
$\phi_{\A}(\lambda)$ of $\A$;  this latter is
a fourteen-degree polynomial with only even degree terms and its constant term is
the square of $\RESULTANT(\A\x^2)$. When $\RESULTANT(\A\x^2)\neq0$,
$\phi_{\A}(\lambda)=0$ if and only if $\lambda$ is an E-eigenvalue of $\A$.

In Section~\ref{sec:4}, by assuming that the North pole $(0,0,1)^\T$ is
the global maximum point of the octupolar potential on the unit sphere,
the admissible region is further reduced.
The boundary of such a reduced admissible region was referred to as the \emph{dome} in \cite{GaV-16}.
By exploring the E-characteristic polynomial $\phi_{\A}(\lambda)$ of $\A$,
we give the algebraic expression for the dome explicitly.

We said already that  two generic octupolar states were identified in  in \cite{GaV-16}
on the two sides of a separatrix surface in parameter space. It should be added here that in  in \cite{GaV-16} both the dome and the separatix were determined by numerical continuation: no closed-form was given for either of them.
In Section~\ref{sec:5}, we also obtain the explicit, algebraic expression for the separatrix.

Finally, some concluding remarks are drawn in Section~\ref{sec:6}.

\section{Preliminaries}\label{sec:2}

The octupolar tensor in liquid crystals is a traceless tensor.
Now, we give the formal definition of traceless tensors and
prove that the traceless property of a symmetric tensor
are invariant under orthogonal transformations \cite{GaV-16}.
For convenience, we denote as $\REAL^{[m,n]}$ the real-valued space of
$m$th order $n$ dimensional symmetric tensors.

\begin{Definition}
  Let $\Ten=[t_{i_1i_2\cdots i_m}]\in\REAL^{[m,n]}$. If
  \begin{equation*}
    \sum_{i=1}^n t_{iii_3\cdots i_m}=0 \qquad\text{ for all } i_3,\ldots,i_m=1,2,\ldots,n,
  \end{equation*}
  then $\Ten$ is called a traceless tensor.
\end{Definition}

\begin{Theorem}
  Let $\Ten=[t_{i_1i_2\cdots i_m}]\in\REAL^{[m,n]}$ be a traceless tensor and
  $Q=[q_{ij}]\in\REAL^{n\times n}$ be an orthogonal matrix.
  We denote by $\Ten Q^m \in\REAL^{[m,n]}$ the symmetric tensor whose elements are
  \begin{equation*}
    [\Ten Q^m]_{i_1i_2\cdots i_m} = \sum_{j_1=1}^n\sum_{j_2=1}^n\cdots\sum_{j_m=1}^n
      t_{j_1j_2\cdots j_m}q_{i_1j_1}q_{i_2j_2}\cdots q_{i_mj_m}.
  \end{equation*}
  Then, $\Ten Q^m$ is also a traceless tensor.
\end{Theorem}
\begin{proof}
  It is straightforward to see that the new tensor $\Ten Q^m$ is real-valued and symmetric.
  Now, we consider its slice matrices.
  As the sum of all matrix eigenvalues, the trace of a symmetric matrix is invariant
  under an orthogonal transformation. Hence, we get
  \begin{equation}\label{0-trace-0}
    \sum_{i=1}^n\sum_{j_1=1}^m\sum_{j_2=1}^m t_{j_1j_2j_3\cdots j_m}q_{ij_1}q_{ij_2}
      = \sum_{i=1}^n t_{iij_3\cdots j_m} = 0
  \end{equation}
  for all $j_3,\ldots,j_m=1,2,\ldots,n$. By some calculations,
  \begin{eqnarray*}
    \sum_{i=1}^n [\Ten Q^m]_{iii_3\cdots i_m}
      &=& \sum_{i=1}^n \sum_{j_1=1}^n\sum_{j_2=1}^n\sum_{j_3=1}^n\cdots\sum_{j_m=1}^n
            t_{j_1j_2j_3\cdots j_m}q_{ij_1}q_{ij_2}q_{i_3j_3}\cdots q_{i_mj_m} \\
      &=& \sum_{j_3=1}^n\cdots\sum_{j_m=1}^n \left(\sum_{i=1}^n \sum_{j_1=1}^n\sum_{j_2=1}^n
            t_{j_1j_2j_3\cdots j_m}q_{ij_1}q_{ij_2}\right) q_{i_3j_3}\cdots q_{i_mj_m} \\
      &=& 0,
  \end{eqnarray*}
  where the last equality is valid because of \eqref{0-trace-0}.
  Hence, the new tensor $\Ten Q^m$ is also traceless.
\end{proof}


By exploting the rotational invariance of traceless tensors,
we now represent the octupolar tensor as
\begin{equation}\label{A-orig.}
    \A=\left[
         \begin{array}{ccc|ccc|ccc}
           a_{111} & a_{112} & a_{113} & a_{112} & a_{122} & a_{123} & a_{113} & a_{123} & a_{133} \\
           a_{112} & a_{122} & a_{123} & a_{122} & a_{222} & a_{223} & a_{123} & a_{223} & a_{233} \\
           a_{113} & a_{123} & a_{133} & a_{123} & a_{223} & a_{233} & a_{133} & a_{233} & a_{333} \\
         \end{array}
       \right]\in\REAL^{[3,3]}
\end{equation}
by choosing a proper Cartesian coordinate system.
The traceless property of $\A$ means that
\begin{equation}\label{traceless}
\left\{\begin{aligned}
  a_{111}+a_{122}+a_{133} = 0, \\
  a_{112}+a_{222}+a_{233} = 0, \\
  a_{113}+a_{223}+a_{333} = 0. \\
\end{aligned}\right.
\end{equation}
Hence, there are seven independent elements in $\A$.
Let $$\alpha_0 = a_{123}, ~\alpha_1 = a_{111}, ~\alpha_2 = a_{222}, ~\alpha_3 = a_{333},
~\beta_1  = a_{122}, ~\beta_2  = a_{233}, ~\beta_3 = a_{113}.$$
Using the traceless property \eqref{traceless}, we convert
\eqref{A-orig.} into
\begin{eqnarray*}
  \lefteqn{ \A = } \\
  && \footnotesize\left[
         \begin{array}{ccc|ccc|ccc}
           \alpha_1 & -\alpha_2-\beta_2 & \beta_3 & -\alpha_2-\beta_2 & \beta_1 & \alpha_0 & \beta_3 & \alpha_0 & -\alpha_1-\beta_1 \\
           -\alpha_2-\beta_2 & \beta_1 & \alpha_0 & \beta_1 & \alpha_2 & -\alpha_3-\beta_3 & \alpha_0 & -\alpha_3-\beta_3 & \beta_2 \\
           \beta_3 & \alpha_0 & -\alpha_1-\beta_1 & \alpha_0 & -\alpha_3-\beta_3 & \beta_2 & -\alpha_1-\beta_1 & \beta_2 & \alpha_3 \\
         \end{array}
    \right].
\end{eqnarray*}
The associated octupolar potential as defined in \cite{GaV-16} is
\begin{eqnarray}\label{poten}
    \Phi(\x) &\equiv& \A\x^3 = \sum_{i=1}^3\sum_{j=1}^3\sum_{k=1}^3 a_{ijk}x_ix_jx_k \\
      &=& \alpha_1x_1^3 + \alpha_2x_2^3 + \alpha_3x_3^3 + 6\alpha_0x_1x_2x_3
         + 3\beta_1x_1x_2^2 + 3\beta_2x_2x_3^2 + 3\beta_3x_1^2x_3 \nonumber\\
      &&{} - 3(\alpha_1+\beta_1)x_1x_3^2 - 3(\alpha_2+\beta_2)x_1^2x_2 - 3(\alpha_3+\beta_3)x_2^2x_3. \nonumber
\end{eqnarray}
On the unit sphere $\SPHERE\equiv\{\x=(x_1,x_2,x_3)^\T : x_1^2+x_2^2+x_3^2=1\}$,
the polynomial $\Phi(\x)$ has at least one maximum point.
Without loss of generality, we assume such a  maximum point being the North pole $(0,0,1)^\T$ and
\begin{equation*}
    \alpha_3 = \Phi(0,0,1) \geq 0.
\end{equation*}

From the spectral theory of tensors \cite{Qi-05}, we know that $\alpha_3$ is a Z-eigenvalue of $\A$
with an associated Z-eigenvector $(0,0,1)^\T$. All Z-eigenvectors and Z-eigenvalues must satisfy the following system:
\begin{equation*}
\left\{\begin{aligned}
    & \alpha_1x_1^2 + 2\alpha_0x_2x_3 + \beta_1x_2^2 + 2\beta_3x_1x_3
      - (\alpha_1+\beta_1)x_3^2 - 2(\alpha_2+\beta_2)x_1x_2 = \lambda x_1, \\
    & \alpha_2x_2^2 + 2\alpha_0x_1x_3 + 2\beta_1x_1x_2 + \beta_2x_3^2
      - (\alpha_2+\beta_2)x_1^2 - 2(\alpha_3+\beta_3)x_2x_3 = \lambda x_2, \\
    & \alpha_3x_3^2 + 2\alpha_0x_1x_2 + 2\beta_2x_2x_3 + \beta_3x_1^2
      - 2(\alpha_1+\beta_1)x_1x_3 - (\alpha_3+\beta_3)x_2^2 = \lambda x_3, \\
    & x_1^2 + x_2^2 + x_3^2 = 1.
\end{aligned}\right.
\end{equation*}
Hence, requiring $(0,0,1)^\T$ to be a solution, we obtain
\begin{equation*}
    \alpha_1+\beta_1=0 \qquad\text{ and }\qquad \beta_2=0.
\end{equation*}
Moreover, because $\Phi(-x_1,0,0)=-\Phi(x_1,0,0)$, we could rotate the Cartesian coordinate system
so that $\Phi(1,0,0)=0$ and we get
\begin{equation*}
    \alpha_1=0.
\end{equation*}
Now, the octupolar tensor in \eqref{A-orig.} reduces to
\begin{equation*}
    \A = \left[
         \begin{array}{ccc|ccc|ccc}
           0 & -\alpha_2 & \beta_3 & -\alpha_2 & 0 & \alpha_0 & \beta_3 & \alpha_0 & 0 \\
           -\alpha_2 & 0 & \alpha_0 & 0 & \alpha_2 & -\alpha_3-\beta_3 & \alpha_0 & -\alpha_3-\beta_3 & 0 \\
           \beta_3 & \alpha_0 & 0 & \alpha_0 & -\alpha_3-\beta_3 & 0 & 0 & 0 & \alpha_3 \\
         \end{array}
    \right],
\end{equation*}
which features four independent elements, namely, $\alpha_0$, $\alpha_2$, $\alpha_3$, and $\beta_3$.
Correspondingly, the octupolar potential \eqref{poten} is
\begin{equation*}
    \Phi(\x;\alpha_0,\alpha_2,\alpha_3,\beta_3) = \alpha_2x_2^3 + \alpha_3x_3^3 + 6\alpha_0x_1x_2x_3
         + 3\beta_3x_1^2x_3 - 3\alpha_2x_1^2x_2 - 3(\alpha_3+\beta_3)x_2^2x_3
\end{equation*}
for all $\x\in\SPHERE$. Without loss of generality, we can assume
\begin{equation*}
    \alpha_2 \geq 0
\end{equation*}
as a consequence of the following proposition.

\begin{Proposition}\label{alpha2-sym}
  For the octupolar potential \eqref{poten-reduced}, we have
  \begin{equation*}
    \Phi(x_1,x_2,x_3;\alpha_0,\alpha_2,\alpha_3,\beta_3)
      = \Phi(x_1,-x_2,x_3;-\alpha_0,-\alpha_2,\alpha_3,\beta_3).
  \end{equation*}
\end{Proposition}

Next, we turn to the assumption that the North pole $(0,0,1)^\T$ is
the maximum point of the octupolar potential with value $\alpha_3$.

\begin{Theorem}\label{Th-admission}
  Suppose that the North pole $(0,0,1)^\T$ is the maximum point of the octupolar potential $\Phi(\x)$ on $\SPHERE$.
  Then, we have that
  \begin{equation}\label{Admiss}
    3\alpha_3^2 - 4\alpha_3\beta_3 - 4\beta_3^2 - 4\alpha_0^2 \geq 0.
  \end{equation}
  If the strict inequality holds in \eqref{Admiss}, then $(0,0,1)^\T$ is
  a (local) maximum point of $\Phi(\x)$ on $\SPHERE$.
\end{Theorem}
\begin{proof}
  We consider the spherical optimization problem:
  \begin{equation}\label{sph-opt}
  \left\{\begin{aligned}
    \max ~& \Phi(\x)=\A\x^3 \\
    \st ~~& \x^\T\x=1.
  \end{aligned}\right.
  \end{equation}
  Its Lagrangian is
  \begin{equation*}
    L(\x,\lambda) = -\frac{1}{3}\A\x^3 + \frac{\lambda}{2}(\x^\T\x-1).
  \end{equation*}
  The Hessian of the Lagrangian is
  \begin{equation}\label{Hessian}
    \nabla_{\x\x}^2L(\x,\lambda) = \lambda I - 2\A\x,
  \end{equation}
  which is positive semidefinite on the tangent space $\x^{\bot}\equiv\{\y\in\REAL^3:\x^\T\y=0\}$
  if $\x$ is a (local) maximum point of $\Phi(\x)$ on $\SPHERE$ \cite{NoW-06}.
  Let $P \equiv I-\x\x^\T \in\REAL^{3\times3}$ be the projection matrix onto $\x^{\bot}$.
  Then, the matrix $P^\T \nabla_{\x\x}^2L(\x,\lambda) P$ is positive semidefinite.
  By use of the first-order necessary condition,
  \begin{equation}\label{aaaaaa}
    \A\x^2 = \lambda\x \qquad\text{ and }\qquad \x^\T\x=1,
  \end{equation}
  we have that
  \begin{equation}\label{cccccccc}
    P^\T \nabla_{\x\x}^2L(\x,\lambda) P = \lambda(I+\x\x^\T) -2\A\x.
  \end{equation}


  Because the North pole $(0,0,1)^\T$ is a maximum point with
  the associated multiplier $\lambda=\alpha_3$, we arrive at
  \begin{equation*}
    [P^\T \nabla_{\x\x}^2 L(\x,\lambda)P]_{\lambda=\alpha_3,\x=(0,0,1)^\T}
      = \left[
          \begin{array}{ccc}
            \alpha_3-2\beta_3 & -2\alpha_0  & 0 \\
            -2\alpha_0 & 3\alpha_3+2\beta_3 & 0 \\
            0 & 0 & 0\\
          \end{array}
        \right].
  \end{equation*}
  As easily seen, this projected Hessian has eigenvalues
  \begin{equation*}
    \mu_1 = 0, \qquad
    \mu_{2,3} = 2\alpha_3 \pm \sqrt{(\alpha_3+2\beta_3)^2+4\alpha_0^2},
  \end{equation*}
  which are all required to be non-negative \cite{NoW-06}. Hence, we obtain the following inequality
  \begin{equation*}
    3\alpha_3^2 - 4\alpha_3\beta_3 - 4\beta_3^2 - 4\alpha_0^2 \geq 0.
  \end{equation*}
  If the strict inequality holds, i.e., if $3\alpha_3^2 - 4\alpha_3\beta_3 - 4\beta_3^2 - 4\alpha_0^2 > 0$, then
  the North pole $(0,0,1)^\T$ is a (local) maximum point of $\Phi(\x)$ on $\SPHERE$ \cite{NoW-06}.
\end{proof}

We first consider the case that $\alpha_3=0$.
From Theorem \ref{Th-admission}, we know that $\alpha_0=\beta_3=0$.
If $\alpha_2>0$, then $\Phi(0,1,0)=\alpha_2>\alpha_3=\Phi(0,0,1)$.
This contradicts that $(0,0,1)^\T$ is a maximum point. Hence $\alpha_2=0$ and
the octupolar tensor $\A$ is the trivial zero tensor.

In the remainder of this paper, we shall consider the case that $\alpha_3$ is positive.
Without loss of generality, by Proposition \ref{alpha2-sym} and Theorem \ref{Th-admission}, we can choose
\begin{equation}\label{Admissive-region}
    \alpha_3=1, \quad \alpha_2\geq0, \quad\text{ and }\quad \alpha_0^2 + (\beta_3+\tfrac{1}{2})^2 \leq 1.
\end{equation}
Then, the octupolar tensor
\begin{equation}\label{A-4ele.}
    \A(\alpha_0,\beta_3,\alpha_2) = \left[
         \begin{array}{ccc|ccc|ccc}
           0 & -\alpha_2 & \beta_3 & -\alpha_2 & 0 & \alpha_0 & \beta_3 & \alpha_0 & 0 \\
           -\alpha_2 & 0 & \alpha_0 & 0 & \alpha_2 & -1-\beta_3 & \alpha_0 & -1-\beta_3 & 0 \\
           \beta_3 & \alpha_0 & 0 & \alpha_0 & -1-\beta_3 & 0 & 0 & 0 & 1 \\
         \end{array}
    \right]
\end{equation}
has only three independent elements and the associated octupolar potential is given by
\begin{equation}\label{poten-reduced}
    \Phi(\x) = \alpha_2x_2^3 + x_3^3 + 6\alpha_0x_1x_2x_3
         + 3\beta_3x_1^2x_3 - 3\alpha_2x_1^2x_2 - 3(1+\beta_3)x_2^2x_3
\end{equation}
for all $\x\in\SPHERE$.



\section{The E-characteristic polynomial}\label{sec:3}

Qi \cite{Qi-05} introduced E-eigenvalues and E-eigenvectors for a symmetric tensor and showed that they are
invariant under orthonormal coordinate changes.
E-eigenvalues and E-eigenvectors were further studied in \cite{Qi-07,NQWW-07,CaS-13}.
Furthermore, the coefficients of the E-characteristic polynomial of a tensor are
orthonormal invariants of that tensor \cite{LQZ-13}.
If the E-eigenvalues and the associated E-eigenvectors of a real-valued symmetric tensor
are real, we call them the Z-eigenvalues and the Z-eigenvectors of the tensor, respectively.

Using the notion of resultant from algebraic geometry \cite{CLO-04},
we write now compute explicitly the resultant $\RESULTANT(\A\x^2)$, where
\begin{equation*}
    \A\x^2 = \left(\begin{array}{c}
      -2\alpha_2x_1x_2 +2\beta_3x_1x_3 +2\alpha_0x_2x_3 \\
      -\alpha_2x_1^2+\alpha_2x_2^2+2\alpha_0x_1x_3-2(1+\beta_3)x_2x_3 \\
      \beta_3x_1^2 -(1+\beta_3)x_2^2 + x_3^2 +2\alpha_0x_1x_2
    \end{array}\right)
    \equiv \left(\begin{aligned}
      F_3(x_1,x_2,x_3) \\ F_1(x_1,x_2,x_3) \\ F_2(x_1,x_2,x_3)
    \end{aligned}\right).
\end{equation*}
Since $F_1$, $F_2$, and $F_3$ are homogeneous polynomials with degree two in the variables $x_1$, $x_2$, and $x_3$,
the multi-polynomial system
\begin{equation}\label{poly-system-0}
    F_1(x_1,x_2,x_3) = F_2(x_1,x_2,x_3) = F_3(x_1,x_2,x_3) = 0
\end{equation}
has a trivial solution $x_1=x_2=x_3=0$. However, we are interested in its non-trivial solutions and we
assume that the multi-polynomial system \eqref{poly-system-0} has a non-zero common (complex) root.


According to Theorem 2.3 in Chapter 3 of \cite{CLO-04}, there is a unique irreducible polynomial
$\RESULTANT(F_1,F_2,F_3)$ such that $\RESULTANT(x_1^2,x_2^2,x_3^2)=1$
and the system $F_1=F_2=F_3=0$ has a non-trivial solution over $\COMPLEX$ if, and only if,
$\RESULTANT(F_1,F_2,F_3)=0$.
We follow the approach in Chapter 3, $\S4$ of \cite{CLO-04}. Since the degree of each of the $F_i$ is $d_i=2$, we set
\begin{equation*}
    d = \sum_{i=1}^3(d_i-1)+1 = 4.
\end{equation*}
We divide monomials $\x^{{\bm \upsilon}}=x_1^{\upsilon_1}x_2^{\upsilon_2}x_3^{\upsilon_3}$ of
total degree $|{\bm \upsilon}|\equiv \upsilon_1+\upsilon_2+\upsilon_3=d$ into three sets:
\begin{eqnarray*}
  S_1 &=& \{\x^{{\bm \upsilon}}:|{\bm \upsilon}|=d, x_1^2\text{ divides }\x^{{\bm \upsilon}}\}
      = \{x_1^4, x_1^3x_2, x_1^3x_3, x_1^2x_2^2, x_1^2x_2x_3, x_1^2x_3^2\}, \\
  S_2 &=& \{\x^{{\bm \upsilon}}:|{\bm \upsilon}|=d, x_1^2\text{ doesn't divides }\x^{{\bm \upsilon}}
          \text{ but }x_2^2\text{ does}\}
      = \{x_1x_2^3, x_1x_2^2x_3, x_2^4, x_2^3x_3, x_2^2x_3^2\}, \\
  S_3 &=& \{\x^{{\bm \upsilon}}:|{\bm \upsilon}|=d, x_1^2,x_2^2\text{ don't divide }\x^{{\bm \upsilon}}
          \text{ but }x_3^2\text{ does}\}
      = \{x_1x_2x_3^2, x_1x_3^3, x_2x_3^3, x_3^4\}.
\end{eqnarray*}
Clearly, there are ${d+2 \choose 2}=15$ monomials $\x^{{\bm \upsilon}}$ with $|{\bm \upsilon}|=d$
and each lies in one of sets $S_1$, $S_2$, and $S_3$, which are mutually disjoint.
We next write the system of equations
\begin{equation*}
\left\{\begin{aligned}
  \x^{{\bm \upsilon}}/x_1^2 \cdot F_1 = 0 & & \text{ for all }\x^{{\bm \upsilon}}\in S_1, \\
  \x^{{\bm \upsilon}}/x_2^2 \cdot F_2 = 0 & & \text{ for all }\x^{{\bm \upsilon}}\in S_2, \\
  \x^{{\bm \upsilon}}/x_3^2 \cdot F_3 = 0 & & \text{ for all }\x^{{\bm \upsilon}}\in S_3. \\
\end{aligned}\right.
\end{equation*}
Its coefficient matrix,
\begin{eqnarray*}
D &=& \left[
  \begin{array}{cccccccc}
    -\alpha_2 &  & 2\alpha_0 & \alpha_2 & -2(1+\beta_3) &  &  &  \\
     & -\alpha_2 &  &  & 2\alpha_0 &  & \alpha_2 & -2(1+\beta_3) \\
     &  & -\alpha_2 &  &  & 2\alpha_0 &  & \alpha_2 \\
     &  &  & -\alpha_2 &  &  &  & 2\alpha_0  \\
     &  &  &  & -\alpha_2 &  &  & \\
     &  &  &  &  & -\alpha_2 &  & \\
     & \beta_3 &  & 2\alpha_0 &  &  & -1-\beta_3 & \\
     &  & \beta_3 &  & 2\alpha_0 &  &  & -1-\beta_3 \\
     &  &  & \beta_3 &  &  & 2\alpha_0 & \\
     &  &  &  & \beta_3 &  &  & 2\alpha_0 \\
     &  &  &  &  & \beta_3 &  & \\
     &  &  & -2\alpha_2 & 2\beta_3 &  &  & 2\alpha_0 \\
     &  &  &  & -2\alpha_2 & 2\beta_3 &  &  \\
     &  &  &  &  &  &  & -2\alpha_2\\
     &  &  &  &  &  &  & \\
  \end{array}
\right. \\
  &&{} \left.
  \begin{array}{cccccccc}
     &  &  &  &  &  &  &  \\
     &  &  &  &  &  &  &  \\
     &  &  & -2(1+\beta_3) &  &  &  \\
    \alpha_2 & -2(1+\beta_3) &  &  &  &  &  \\
     & \alpha_2 & -2(1+\beta_3) & 2\alpha_0 &  &  &  \\
     &  & \alpha_2 &  & 2\alpha_0 & -2(1+\beta_3) &  \\
     &  &  & 1 &  &  &  \\
     &  &  &  & 1 &  &  \\
    -1-\beta_3 &  & 1 &  &  &  &  \\
     & -1-\beta_3 &  &  &  & 1 &  \\
     &  & -1-\beta_3 & 2\alpha_0 &  &  & 1 \\
     &  &  &  &  &  &  &  \\
     &  &  & 2\alpha_0 &  &  &  \\
     &  & 2\alpha_0 & 2\beta_3 &  &  &  \\
     &  &  & -2\alpha_2 & 2\beta_3 & 2\alpha_0 &  \\
  \end{array}
\right] 
\end{eqnarray*}
in the unknowns $\x^{{\bm \upsilon}}$ with $|{\bm \upsilon}|=d$ is important in the sense that
\begin{equation*}
    \det(D) = \RESULTANT(\A\x^2) \cdot \text{extraneous factor}.
\end{equation*}


Now, we consider to the extraneous factor. A monomial $\x^{{\bm \upsilon}}$ of total degree $d=4$
is reduced if $x_i^{2}$ divides $\x^{{\bm \upsilon}}$ for exactly one $i$.
Let $D'$ be the determinant of the submatrix of $D$ obtained by deleting all rows and columns
corresponding to reduced monomials, i.e.,
\begin{equation*}
    D'=\left[
         \begin{array}{ccc}
           -\alpha_2 &  &  \\
            & -\alpha_2 & \alpha_2 \\
            & \beta_3 & -1-\beta_3 \\
         \end{array}
       \right].
\end{equation*}
By Theorem 4.9 in Chapter 3 of \cite{CLO-04}, to within a sign, the resultant reads as
\begin{eqnarray}\label{Result-Ax2}
    \RESULTANT(\A\x^2) &=& \frac{\det D}{\det D'} \\
      &=& 16\alpha_2^2(48\alpha_0^8\beta_3+4\alpha_0^6(\alpha_2^2+\beta_3(32\beta_3^2+24\beta_3-9))
          +3\alpha_0^4(\alpha_2^2(52\beta_3^2+28\beta_3-1) \nonumber\\&&{} +4\beta_3^2(8\beta_3^3+8\beta_3^2-9\beta_3-9))
          +6\alpha_0^2(\alpha_2^4(4\beta_3+1)-\alpha_2^2\beta_3(14\beta_3^3+36\beta_3^2+35\beta_3 \nonumber\\&&{} +10)
          -2\beta_3^3(\beta_3+1)^2(8\beta_3+9))
          +(\alpha_2^2-4(\beta_3+1)^3)(\alpha_2^2-\beta_3^2(2\beta_3+3))^2).  \nonumber
\end{eqnarray}

\begin{Theorem}[\cite{CLO-04}]\label{Th-Ax^2}
  There exists a vector $\x\neq\vt{0}$ such that $\A\x^2=\vt{0}$ if, and only if,
  $\RESULTANT(\A\x^2)=0$, where the formula for $\RESULTANT(\A\x^2)$
  is given by \eqref{Result-Ax2}.
\end{Theorem}

By the same approach, we compute the E-characteristic polynomial $\phi_{\A}(\lambda)$
of the octupolar tensor \eqref{A-4ele.}, which is a resultant of
the following system of homogeneous polynomial equations
\begin{equation}\label{poly-system}
\left\{\begin{array}{rcl}
  x_1^2+x_2^2+x_3^2 - x_0^2 &=& 0, \\
  -\alpha_2x_1^2+\alpha_2x_2^2+2\alpha_0x_1x_3-2(1+\beta_3)x_2x_3 - \lambda x_0x_2 &=& 0, \\
  \beta_3x_1^2 -(1+\beta_3)x_2^2 + x_3^2 +2\alpha_0x_1x_2 - \lambda x_0x_3 &=& 0, \\
  -2\alpha_2x_1x_2 +2\beta_3x_1x_3 +2\alpha_0x_2x_3 - \lambda x_0x_1 &=& 0.
\end{array}\right.
\end{equation}
Using the software Mathematica, we obtain the E-characteristic polynomial $\phi_{\A}(\lambda)$.

\begin{Theorem}\label{Th-E_char}
  The E-characteristic polynomial of the octupolar tensor \eqref{A-4ele.} is
  \begin{equation*}
    \phi_{\A}(\lambda) = (\lambda^2-1)\sum_{i=0}^6c_{2i}\lambda^{2i},
  \end{equation*}
  where
  \begin{eqnarray*}
  c_0 &=& 256\alpha_2^4(48\alpha_0^8\beta_3+4\alpha_0^6(\alpha_2^2+\beta_3(32\beta_3^2+24\beta_3-9))
    +3\alpha_0^4(\alpha_2^2(52\beta_3^2+28\beta_3-1)+4\beta_3^2(8\beta_3^3 \\&&{} +8\beta_3^2-9\beta_3-9))
    +6\alpha_0^2(\alpha_2^4(4\beta_3+1)-\alpha_2^2\beta_3(14\beta_3^3+36\beta_3^2+35\beta_3+10)
      -2\beta_3^3(\beta_3 \\&&{} +1)^2(8\beta_3+9))
    +(\alpha_2^2-4(\beta_3+1)^3)(\alpha_2^2-\beta_3^2(2\beta_3+3))^2)^2,
  \end{eqnarray*}
  and
  \begin{eqnarray*}
  c_{12} &=& 82944 \alpha_0^{10}-11520 \alpha_0^8 (\alpha_2^2-36 \beta_3^2-36 \beta_3+1)-320 \alpha
   _0^6 (2 \alpha_2^2 (72 \beta_3^2-1053 \beta_3-577)+73 \alpha
   _2^4 \\&&{} -2592 \beta_3^4-5184 \beta_3^3-2448 \beta_3^2+144 \beta_3+73)-240
   \alpha_0^4 (\alpha_2^6 -\alpha_2^4 (1583 \beta_3^2+1208 \beta
   _3+922) \\&&{} +\alpha_2^2 (288 \beta_3^4-4424 \beta_3^3-7328 \beta_3^2-116
   \beta_3+1203)-(2 \beta_3+1){}^2 (864 \beta_3^4
   +1728 \beta_3^3+576 \beta_3^2 \\&&{} -288 \beta_3-1))+60 \alpha_0^2
   (32 \alpha_2^8+ \alpha_2^6 (-8 \beta_3^2+1992 \beta_3+678)-\alpha_2^4
   (6168 \beta_3^4+13336 \beta_3^3 \\&&{} +5042 \beta_3^2-4376 \beta_3+1083)-2
   \alpha_2^2 (384 \beta_3^6-848 \beta_3^5-4080 \beta_3^4-80 \beta_3^3+4580
   \beta_3^2+437 \beta_3 \\&&{} -714)+8 (2 \beta_3+1){}^4
   (54 \beta_3^4+108 \beta_3^3+21 \beta_3^2-33 \beta
   _3+4))+(\alpha_2^2+(3 \beta_3+4)^2)
   (16 \alpha_2^4  \\&&{} + \alpha_2^2 (-12 \beta_3^2-132 \beta_3+37)+4
   (2 \beta_3+1){}^3 (3 \beta_3-1))^2.
  \end{eqnarray*}
\end{Theorem}

The E-characteristic polynomial $\phi_{\A}(\lambda)$ is a  polynomial of degree $14$,
with no odd-degree terms.
By comparing the expression of $c_{12}$ and (\ref{Result-Ax2}), we have the following corollary:

\begin{Corollary}
  The constant term of the E-characteristic polynomial $\phi_{\A}(\lambda)$ is
  \begin{equation*}
    c_0 = \left[\RESULTANT(\A\x^2)\right]^2.
  \end{equation*}
\end{Corollary}

This corollary is in agreement with Theorem 3.5 of \cite{LQZ-13} in this case.

\begin{Theorem}\label{Th-e-charPoly}
  Suppose that $\A\x^2=\vt{0}$ has only the trivial solution $\x=\vt{0}$.
  Then, the system \eqref{poly-system} has a non-trivial solution $(x_0,\ldots,x_3)$ over $\COMPLEX$
  if, and only if, $\phi_{\A}(\lambda)=0$.
\end{Theorem}
\begin{proof}
  Because $\A\x^2=\vt{0}$ has only the trivial solution $\x=\vt{0}$, we know that the system
  of $\A\x^2 = \vt{0}$ and $\x^\T\x = 0$ has only a zero solution.
  By Proposition 2.6 in Chapter 3 of \cite{CLO-04}, we obtain the desired conclusion.
\end{proof}

\begin{Corollary}
  If $\RESULTANT(\A\x^2)\neq 0$, then all E-eigenvalues of the octupolar tensor $\A$ are non-zero.
\end{Corollary}
\begin{proof}
  Since $\RESULTANT(\A\x^2)\neq 0$, the system $\A\x^2=\vt{0}$ has only
  the trivial solution by Theorem \ref{Th-Ax^2}.
  Let us compute $\phi_{\A}(0)=[\RESULTANT(\A\x^2)]^2 >0$. By Theorem \ref{Th-e-charPoly},
  $\lambda=0$ is not an E-eigenvalue of $\A$.
\end{proof}

\section{Dome: the reduced admissible region}\label{sec:4}

Assume that the North pole $(0,0,1)^\T$ is the global maximum point of
the octupolar potential $\Phi(\x)$ on the unit sphere $\SPHERE$.
That is, we assume that the octupolar tensor $\A$ has the largest Z-eigenvalue
$\lambda=1$ with an associated Z-eigenvector $(0,0,1)^\T$.
In the admissible region \eqref{Admissive-region}, there is a reduced region
such that the maximal Z-eigenvalue of $\A$ is $1$.
The boundary of this reduced admissible region is called the \emph{dome} \cite{GaV-16}:
its apex is at $\alpha_0=0,\beta_3=-\frac{1}{2},\alpha_2=\frac{\sqrt{2}}{2}$,
and it meets the plane $\alpha_2=0$ along the circle $\alpha_0^2+\beta_3^2+\beta_3=0$.

Now, we are in a position to give an explicit formula for the dome.
We consider the E-characteristic polynomial $\phi_{\A}(\lambda)$ in Theorem \ref{Th-E_char}.
Clearly, $\lambda=1$ is a root of $\phi_{\A}(\lambda)$.
Since the dome is the locus where the maximal Z-eigenvalue is $\lambda=1$,
we substitute $\lambda=1$ into $\phi_{\A}(\lambda)/(\lambda^2-1)=0$
and we obtain the following equation
\begin{equation}\label{dome-11}
    c_1(\alpha_0,\beta_3,\alpha_2)^3 \cdot c_2(\alpha_0,\beta_3,\alpha_2) \cdot c_3(\alpha_0,\beta_3,\alpha_2) = 0,
\end{equation}
where
\begin{equation*}
    c_1(\alpha_0,\beta_3,\alpha_2) = 3 - 4\alpha_0^2 - 4\beta_3^2 - 4\beta_3,
\end{equation*}
\begin{equation*}
    c_2(\alpha_0,\beta_3,\alpha_2) = 64\alpha_2^4 - 16\alpha_2^2(1+2\beta_3)(-12\alpha_0^2+(1+2\beta_3)^2)
      + (4\alpha_0^2+(1+2\beta_3)^2)^3,
\end{equation*}
and
\begin{eqnarray}\label{dome-c3}
    \lefteqn{c_3(\alpha_0,\beta_3,\alpha_2) = \alpha_2^6(2\beta_3-1)((2\beta_3+5)^2-12\alpha_0^2)
      +\alpha_2^4(-48\alpha_0^4(3\beta_3^2-1)
        +12\alpha_0^2(8\beta_3^4}  \nonumber\\
   &&{} +24\beta_3^3+26\beta_3^2-4\beta_3-11)-16\beta_3^6-96\beta_3^5-168\beta_3^4-
     72\beta_3^3-21\beta_3^2-24\beta_3+40) \nonumber\\
   &&{} +8\alpha_2^2(8\alpha_0^6+6\alpha_0^4(4\beta_3^2-2\beta_3-5)
     +3\alpha_0^2(8\beta_3^4+8\beta_3^3-12\beta_3^2-3\beta_3+6)+8\beta_3^6 \nonumber\\
   &&{} +36\beta_3^5+42\beta_3^4+3\beta_3^3-9\beta_3^2-2)
   -16(\alpha_0^2+\beta_3^2+\beta_3)^2(4\alpha_0^2+4\beta_3^2+4\beta_3-3).
\end{eqnarray}
Because of \eqref{Admissive-region}, we know that $c_1(\alpha_0,\beta_3,\alpha_2)\geq0$
and the equality holds on the boundary of the admissible region.
Hence, $c_1(\alpha_0,\beta_3,\alpha_2)=0$ is a trivial solution of \eqref{dome-11}.

As for $c_2(\alpha_0,\beta_3,\alpha_2)$, this is a quadratic function in $\alpha_2^2$ which attains its minimum value $4\alpha_0^2(4\alpha_0^2-3(1+2\beta_3)^2)^2\geq0$.
If $\alpha_0=0$, then $c_2(\alpha_0,\beta_3,\alpha_2)=(-8\alpha_2^2+(1+2\beta_3)^3)^2$. Hence,
when
\begin{equation}\label{dome-line-one}
    \alpha_0=0 \qquad\text{ and }\qquad
    8\alpha_2^2 - (1+2\beta_3)^3 = 0,
\end{equation}
we have $c_2(\alpha_0,\beta_3,\alpha_2)=0$.
If $4\alpha_0^2-3(1+2\beta_3)^2=0$, then $c_2=64(\alpha_2^2+(1+2\beta_3)^3)^2$.
Hence, when
\begin{equation}\label{dome-line-two}
    4\alpha_0^2-3(1+2\beta_3)^2=0 \qquad\text{ and }\qquad
    \alpha_2^2 + (1+2\beta_3)^3=0,
\end{equation}
we also have $c_2(\alpha_0,\beta_3,\alpha_2)=0$.
However, under either \eqref{dome-line-one} or \eqref{dome-line-two},
there are two E-eigenvectors corresponding to the E-eigenvalue $1$.
The first one is the North pole $(0,0,1)^\T$
and the other one is always a complex vector according to  by our direct computations. Hence, we omit these two lines.

Then, we turn attention to the equation $c_3(\alpha_0,\beta_3,\alpha_2)=0$, which
 has multiple roots in $\alpha_2$ for fixed $\alpha_0$ and $\beta_3$.
For example, when $\alpha_0=0$ and $\beta_3=-0.8$,
$\alpha_2^{(1)}=\alpha_2^{(2)}=\frac{2}{\sqrt{17}}\approx0.4851$
and $\alpha_2^{(3)}=\frac{4\sqrt{7}}{5\sqrt{5}}\approx0.9466$ are roots of the equation, whereas
when $\alpha_0=0.1$ and $\beta_3=-0.8$, the roots are $\alpha_2^{(1)}=0.3765$, $\alpha_2^{(2)}=0.5862$,
and $\alpha_2^{(3)}=9459$.
Which value of $\alpha_2$ then describes the dome?
If the largest Z-eigenvalue of $\A(\alpha_0,\beta_3,\alpha_2)$ is larger that $1$, then
the triple $(\alpha_0,\beta_3,\alpha_2)$ is above the dome.
By direct numerical explorations, we found that, for given $\alpha_0$ and $\beta_3$,
the dome lies on the smallest non-negative value of $\alpha_2$ such that $c_3(\alpha_0,\beta_3,\alpha_2)=0$, i.e.,
\begin{equation}\label{formula-dome}
    \alpha_2(\alpha_0,\beta_3) = \min\{\tilde{\alpha}_2\geq 0: c_3(\alpha_0,\beta_3,\tilde{\alpha}_2)=0\}
    \qquad\text{ for }\alpha_0^2+\beta_3^2+\beta_3 \leq 0.
\end{equation}
The contour profile of the dome as given by \eqref{formula-dome} is illustrated in Figure \ref{Fig-dome}.
\begin{figure}[!btp]
  \centering\includegraphics[width=.7\textwidth]{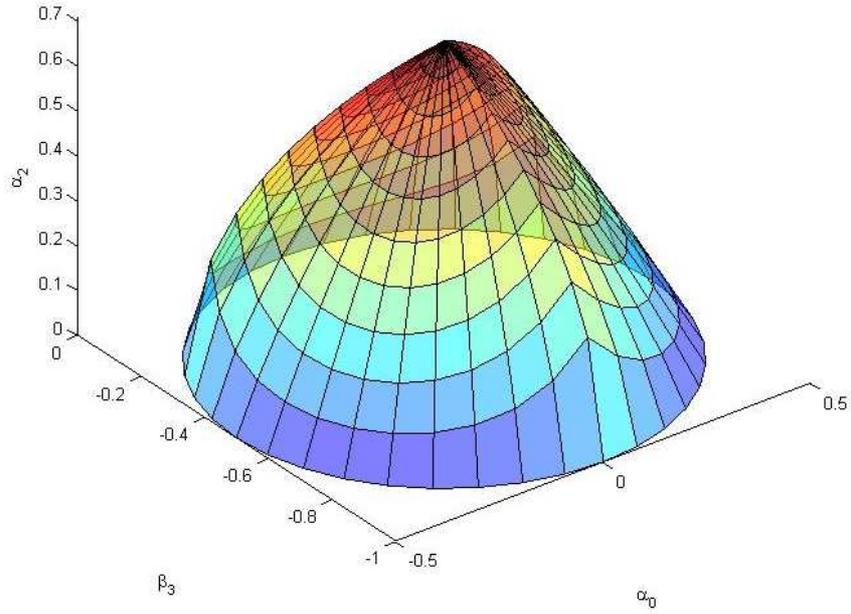}
  \caption{The dome that bounds the reduced admissible region as represented by \eqref{formula-dome}.}\label{Fig-dome}
\end{figure}

Finally, we say more on the apex and the base of the dome.
At the apex of the dome $(\alpha_0,\beta_3,\alpha_2)=(0,-\tfrac{1}{2},\frac{\sqrt{2}}{2})$ and
the E-characteristic polynomial of the octupolar tensor $\A(0,-\tfrac{1}{2},\frac{\sqrt{2}}{2})$ is
$
    \phi_{\A}(\lambda) = 19683 \lambda^6 (\lambda^2-1)^4.
$
There is a quadruple root $\lambda^2=1$ of such a tensor corresponding to
four Z-eigenvectors:
\begin{equation*}
    \x^{(1)}=(0,0,1)^\T, ~~\x^{(2)}=\left(0,\tfrac{2\sqrt{2}}{3},-\tfrac{1}{3}\right)^\T,~~
    \x^{(3)}=\left(\tfrac{\sqrt{6}}{3},-\tfrac{\sqrt{2}}{3},-\tfrac{1}{3}\right)^\T, ~~
    \x^{(4)}=\left(-\tfrac{\sqrt{6}}{3},-\tfrac{\sqrt{2}}{3},-\tfrac{1}{3}\right)^\T.
\end{equation*}
The polar plots of the octupolar potential $\Phi(\x)$
is illustrated in Figure \ref{Fig-poten}(a).
\begin{figure}[!bt]
\begin{center}
\begin{tabular}{ccc}
  \includegraphics[width=.45\textwidth]{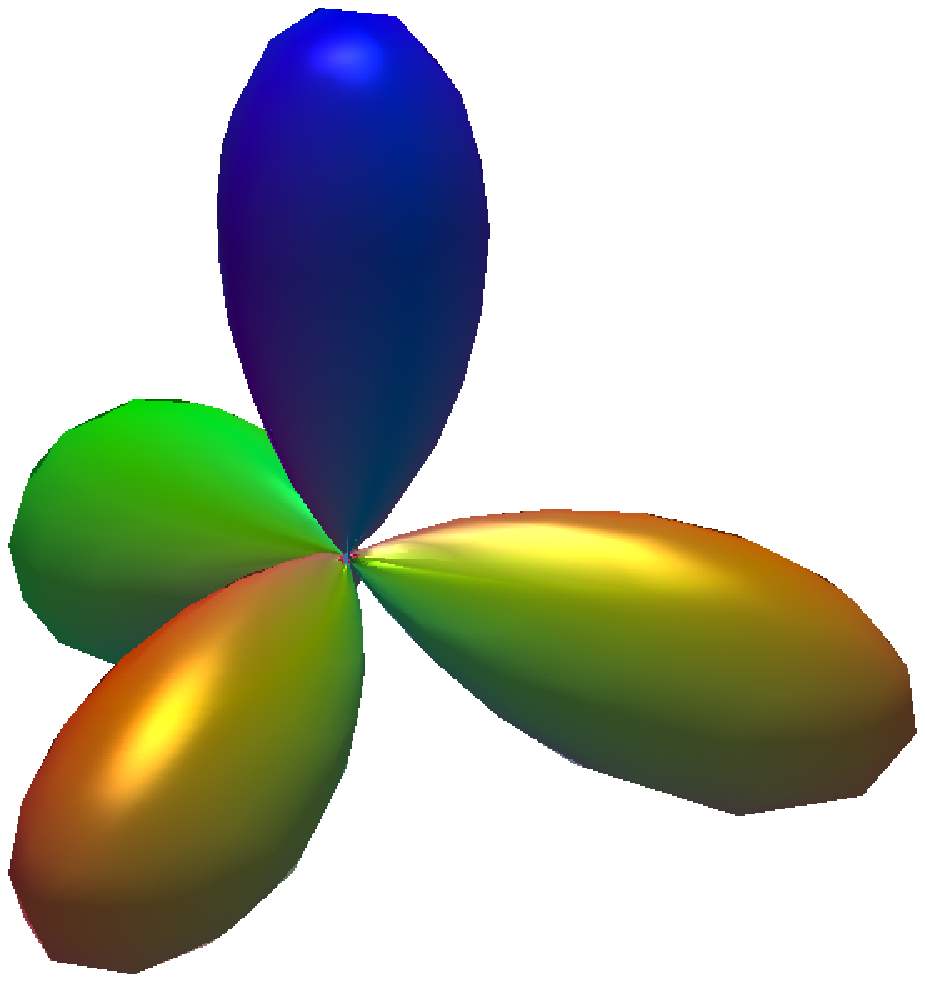} & &
  \includegraphics[width=.45\textwidth]{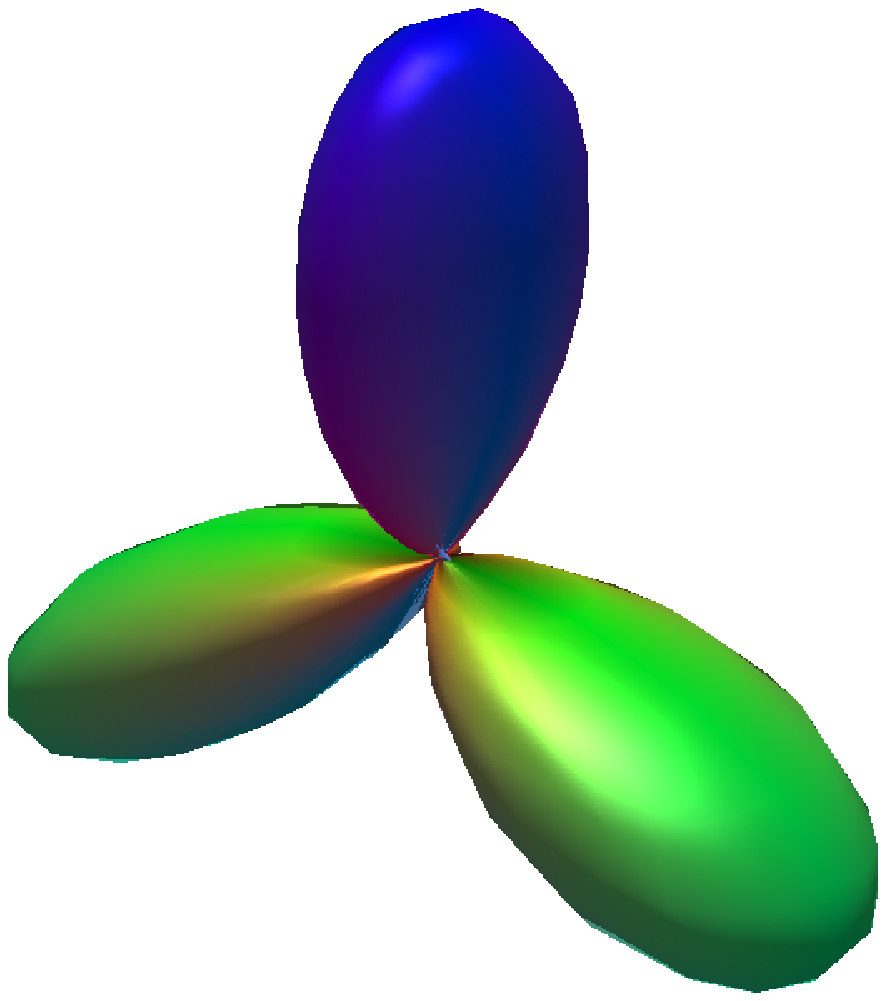} \\
  (a) $\A(0,-\tfrac{1}{2},\frac{\sqrt{2}}{2})$ & \qquad & (b) $\A(0,0,0)$  \\
\end{tabular}
\end{center}
  \caption{The two typical polar plots of the octupolar potential.}\label{Fig-poten}
\end{figure}

At the base of the dome, $\alpha_2=0$ and $\alpha_0^2+\beta_3^2+\beta_3=0$, which represents a circle of center in $\alpha_0=0$, $\beta_3=\frac12$ and radius $\frac12$; there,
the E-characteristic polynomial reduces to
$\phi_{\A}(\lambda)= -64\lambda^8(\lambda^2-1)^6$.
Hence, $\lambda^2=1$ is a triple root of $\phi_{\A}(\lambda)$.
Specifically, $\A(0,0,0)$ has four Z-eigenvectors, namely,
\begin{equation*}
    \x^{(1)}=(0,0,1)^\T, ~~\x^{(2)}=\left(0,\tfrac{\sqrt{3}}{2},-\tfrac{1}{2}\right)^\T,~~
    \x^{(3)}=\left(0,-\tfrac{\sqrt{3}}{2},-\tfrac{1}{2}\right)^\T.
\end{equation*}
corresponding to the positive Z-eigenvalue $\lambda=1$. In this case,
the polar plot of the octupolar potential $\Phi(\x)$
is illustrated in Figure \ref{Fig-poten}(b).
The two plots in Figure~\ref{Fig-poten} illustrate the typical appearance of the octupolar potential in the two generic orientational states of generalized liquid crystals described  by an octupolar order tensor.

\section{Separatrix}\label{sec:5}

Gaeta and Virga \cite{GaV-16} showed by numerical continuation that there is a separatrix surface between
the two different generic states of the octupolar potential $\Phi$:
in one generic state, $\Phi$ has four maxima and three (positive) saddles; in the other generic state,
$\Phi$ has three maxima and two (positive) saddles.\footnote{Since $\Phi(\mathbf{x})$ is odd in $\mathbf{x}$, both maxima (and positive saddles) of the octupolar potential are accompanied by an equal number of minima (and negative saddles) in the antipodal positions on the unit sphere.}
Here, we determine explicitly the separatrix.

We recall the spherical optimization problem \eqref{sph-opt}.
When passing through the separatrix, one maxima $\lambda=\Phi(\x)$ of
the octupolar potential $\Phi(\x)$ on the unit sphere $\SPHERE$ disappears.
Hence, the Hessian $\nabla_{\x\x}^2L(\x,\lambda)$ in \eqref{Hessian} of
the Lagrangian \eqref{sph-opt} is singular on the tangent space $\x^{\bot}$,
i.e., the projected Hessian $P^\T \nabla_{\x\x}^2L(\x,\lambda) P$
in \eqref{cccccccc} has two zero eigenvalues. Clearly, $\mu_1=0$ is an eigenvalue of
the projected Hessian \eqref{cccccccc} with the associated eigenvector $\x$.
Suppose that $\mu_2$ and $\mu_3$ are the other eigenvalues of the projected Hessian \eqref{cccccccc}.
Let $\sigma \equiv \mu_2\mu_3 = \mu_1\mu_2+\mu_1\mu_3+\mu_2\mu_3$.
Then, by linear algebra, $\sigma$ equals the sum of all $2$-by-$2$ principal minors of the projected Hessian.
Using \eqref{cccccccc}, we have that
\begin{eqnarray}\label{sepa-2}
  \sigma = 7\lambda^2 -4\left((\alpha_0^2+\alpha_2^2+\beta_3^2)x_1^2
    + (\alpha_0^2+\alpha_2^2+(\beta_3+1)^2)x_2^2
    + (\alpha_0^2+\beta_3^2+\beta_3+1)x_3^2\right.  \nonumber\\
  {} \left.- 2\alpha_0x_1x_2 - 2\alpha_0\alpha_2x_1x_3 - \alpha_2(2\beta_3+1)x_2x_3\right) = 0.
\end{eqnarray}
Moreover, if $\lambda\neq0$ and $\x\neq0$ satisfy $\A\x^2=\lambda\x$, then
$\frac{\lambda}{\|\x\|}$ and $\frac{1}{\|\x\|}\x$ satisfy \eqref{aaaaaa}.
Hence, we could omit the spherical constraint $\x^\T\x=1$ temporarily
and just consider the system of homogeneous polynomial equations \eqref{sepa-2} and
\begin{equation}\label{sepa-1}
\left\{\begin{array}{rcl}
   -2\alpha_2x_1x_2 +2\beta_3x_1x_3 +2\alpha_0x_2x_3               - \lambda x_1 &=& 0, \\
   -\alpha_2x_1^2+\alpha_2x_2^2+2\alpha_0x_1x_3-2(1+\beta_3)x_2x_3 - \lambda x_2 &=& 0, \\
   \beta_3x_1^2 -(1+\beta_3)x_2^2 +x_3^2 +2\alpha_0x_1x_2          - \lambda x_3 &=& 0.
\end{array}\right.
\end{equation}

Using the approach introduced in Section 3, we obtain the resultant of \eqref{sepa-2} and \eqref{sepa-1}
\begin{equation}\label{separatrix}
    \mathrm{Separatrix} : 1792(4\alpha_0^2+4\beta_3^2+4\beta_3-3)^2\sum_{i=0}^{8}d_{2i}(\alpha_0,\beta_3)\alpha_2^{2i}=0,
\end{equation}
where
\begin{eqnarray*}
  d_{16} &=& 27(-16\alpha_0^4-8\alpha_0^2(4\beta_3^2-44\beta_3+13)-(2\beta_3-1)(2\beta_3+7)^3), \\
  d_{14} &=& -54(128\alpha_0^6-16\alpha_0^4(48\beta_3^2+78\beta_3-29)
      +16\alpha_0^2(72\beta_3^4+124\beta_3^3+190\beta_3^2-101\beta_3-69) \\
      &&{}+(2\beta_3+7)^2(40\beta_3^3+44\beta_3^2+62\beta_3-47)), \\
  d_{12} &=& -9(4096\alpha_0^8-128\alpha_0^6(277\beta_3^2-92\beta_3-55)
      +48\alpha_0^4(152\beta_3^4+1944\beta_3^3-7094\beta_3^2-1548\beta_3 \\
      &&{} +53)  +8\alpha_0^2(5648\beta_3^6+18912\beta_3^5+60408\beta_3^4+115368\beta_3^3
      +86625\beta_3^2-44964\beta_3-18410) \\
      &&{} -1664\beta_3^8-22400\beta_3^7-124064\beta_3^6-377088\beta_3^5-624840\beta_3^4
      -383256\beta_3^3+109994\beta_3^2 \\
      &&{} +181940\beta_3-17605), \\
  d_{10} &=& -2(22528\alpha_0^{10}+256\alpha_0^8(800\beta_3^2+3620\beta_3+599)
      +64\alpha_0^6(5440\beta_3^4-195290\beta_3^3-97221\beta_3^2 \\
      &&{} -44476\beta_3+8375)
      +16\alpha_0^4(12800\beta_3^6+1073640\beta_3^5+2832444\beta_3^4+2369838\beta_3^3
      -242151\beta_3^2 \\
      &&{} -492540\beta_3-270455)
      +4\alpha_0^2(17920\beta_3^8-1188320\beta_3^7-6499376\beta_3^6-13648368\beta_3^5 \\
      &&{} -10198728\beta_3^4+1289514\beta_3^3+3579185\beta_3^2+123260\beta_3+206555)
      +32768\beta_3^{10}+483200\beta_3^9 \\
      &&{} +3111744\beta_3^8+10647136\beta_3^7
      +19890064\beta_3^6+19640424\beta_3^5+5479324\beta_3^4-6109790\beta_3^3 \\
      &&{} -3422445\beta_3^2+504920\beta_3+3560), \\
\end{eqnarray*}
\begin{eqnarray*}
  d_{8} &=& 5(40960\alpha_0^{12}-12288\alpha_0^{10}(97\beta_3^2+88\beta_3+44)
      +256\alpha_0^8(39921\beta_3^4+34176\beta_3^3+42870\beta_3^2 \\
      &&{} +12132\beta_3+1667)
      -128\alpha_0^6(141080\beta_3^6+419208\beta_3^5+389430\beta_3^4+82228\beta_3^3
      -15613\beta_3^2 \\
      &&{} -100402\beta_3-37063)
      +48\alpha_0^4(212768\beta_3^8+1023680\beta_3^7+1963504\beta_3^6+1378192\beta_3^5 \\
      &&{} -304390\beta_3^4-508976\beta_3^3+63582\beta_3^2-57076\beta_3-52349)
      -8\alpha_0^2(149376\beta_3^{10}+1199488\beta_3^9 \\
      &&{} +4718496\beta_3^8+9599232\beta_3^7
      +9822584\beta_3^6+3227448\beta_3^5-2548818\beta_3^4-2029036\beta_3^3 \\
      &&{} -53961\beta_3^2
      +37902\beta_3-40196)+(2\beta_3+1)^2(10304\beta_3^{10}+154304\beta_3^9+911472\beta_3^8 \\
      &&{} +2786464\beta_3^7+4828732\beta_3^6+3895212\beta_3^5+22345\beta_3^4-1558688\beta_3^3
      -352512\beta_3^2 \\
      &&{} +133184\beta_3-7840)), \\
  d_{6} &=& 16 (28672 \alpha _0^{14}-512 \alpha _0^{12} (688 \beta _3^2+1102 \beta _3+941)-128 \alpha
   _0^{10} (9696 \beta _3^4-40380 \beta _3^3-33951 \beta _3^2 \\&&{} -20148 \beta
   _3-11743)-160 \alpha _0^8 (5632 \beta _3^6-26064 \beta _3^5-7644 \beta
   _3^4+35134 \beta _3^3-57181 \beta _3^2 \\&&{} -51958 \beta _3-5454)+40 \alpha _0^6
   (18944 \beta _3^8-103552 \beta _3^7-737312 \beta _3^6-1217152 \beta
   _3^5-320576 \beta _3^4 \\&&{} +504962 \beta _3^3+120149 \beta _3^2-112824 \beta
   _3-29175)+2 \alpha _0^4 (577536 \beta _3^{10}+1111040 \beta
   _3^9-2474880 \beta _3^8 \\&&{} -6705600 \beta _3^7-341600 \beta _3^6+9137976 \beta
   _3^5+5840100 \beta _3^4-884330 \beta _3^3-684765 \beta _3^2+374580 \beta
   _3 \\&&{} +132449)+\alpha _0^2 (2 \beta _3+1){}^2 (80896 \beta
   _3^{10}+999808 \beta _3^9+3452640 \beta _3^8+5398208 \beta _3^7+3717992 \beta
   _3^6 \\&&{} -367068 \beta _3^5-2064016 \beta _3^4-746875 \beta _3^3+150774 \beta
   _3^2+30796 \beta _3-25928)-(2 \beta
   _3+1){}^4 (2048 \beta _3^{10} \\&&{} +25824 \beta _3^9+135752 \beta _3^8+385692
   \beta _3^7+535154 \beta _3^6+253167 \beta _3^5-114083 \beta _3^4-118464 \beta
   _3^3 \\&&{} -4364 \beta _3^2+7632 \beta _3-656)), \\
  d_4 &=& 16 (-32768 \alpha _0^{16}+ 2048 \alpha _0^{14} (241 \beta _3^2-284 \beta _3-83)+256 \alpha
   _0^{12} (9208 \beta _3^4-5384 \beta _3^3+6390 \beta _3^2 \\&&{} +25496 \beta
   _3+8589)+128 \alpha _0^{10} (25680 \beta _3^6-32160 \beta _3^5+7368
   \beta _3^4+257000 \beta _3^3+212292 \beta _3^2 \\&&{} +23286 \beta _3-10209)+80
   \alpha _0^8 (8064 \beta _3^8-169344 \beta _3^7-304224 \beta _3^6+311872 \beta
   _3^5+774736 \beta _3^4 \\&&{} +306928 \beta _3^3-61620 \beta _3^2-34824 \beta
   _3+1209)-8 \alpha _0^6 (281856 \beta _3^{10}+2529280 \beta _3^9+6835200
   \beta _3^8 \\&&{} +6572800 \beta _3^7+316800 \beta _3^6-2303424 \beta _3^5-174400 \beta
   _3^4+593920 \beta _3^3+124725 \beta _3^2-2310 \beta _3 \\&&{} +6019)-2 \alpha _0^4
   (2 \beta _3+1){}^2 (230144 \beta _3^{10}+1285120 \beta
   _3^9+3244032 \beta _3^8+4304128 \beta _3^7+2583584 \beta _3^6 \\&&{} +28128 \beta
   _3^5-669240 \beta _3^4-261024 \beta _3^3+369 \beta _3^2+30952 \beta
   _3-3232)-4 \alpha _0^2 (2 \beta _3+1){}^4 (5408 \beta
   _3^{10} \\&&{} +10240 \beta _3^9-22272 \beta _3^8-72224 \beta _3^7-83578 \beta _3^6-75384
   \beta _3^5-40635 \beta _3^4+8889 \beta _3^3+10338 \beta _3^2 \\&&{} -2444 \beta
   _3-184)+2 (\beta _3+1){}^2 (2 \beta
   _3+1){}^6 (400 \beta _3^8+4000 \beta _3^7+15408 \beta _3^6+16240 \beta
   _3^5-2449 \beta _3^4 \\&&{} -6128 \beta _3^3+104 \beta _3^2+272 \beta _3-24)), \\
  d_2 &=& -256 (\alpha _0^2+\beta _3^2+\beta _3){}^2 (4 \alpha _0^2+(2
   \beta _3+1){}^2){}^3 (64 \alpha _0^8+8 \alpha _0^6 (32 \beta _3^2-112 \beta
   _3-81)+2 \alpha _0^4 (192 \beta _3^4 \\&&{} -576 \beta _3^3-1356 \beta _3^2-78
   \beta _3+245)+\alpha _0^2 (256 \beta _3^6+384 \beta _3^5-408 \beta
   _3^4-136 \beta _3^3+764 \beta _3^2+215 \beta _3 \\&&{} -62)+(2
   \beta _3+1){}^2 (16 \beta _3^6+144 \beta _3^5+266 \beta _3^4+87 \beta
   _3^3-89 \beta _3^2-32 \beta _3+6)), \\
  d_0 &=& 256 (\alpha _0^2+\beta _3^2+\beta _3){}^4 (4 \alpha _0^2+4 \beta
   _3^2+4 \beta _3-3) (4 \alpha _0^2+(2 \beta
   _3+1){}^2){}^5.
\end{eqnarray*}

\begin{figure}[!btp]
  \centering\includegraphics[width=.7\textwidth]{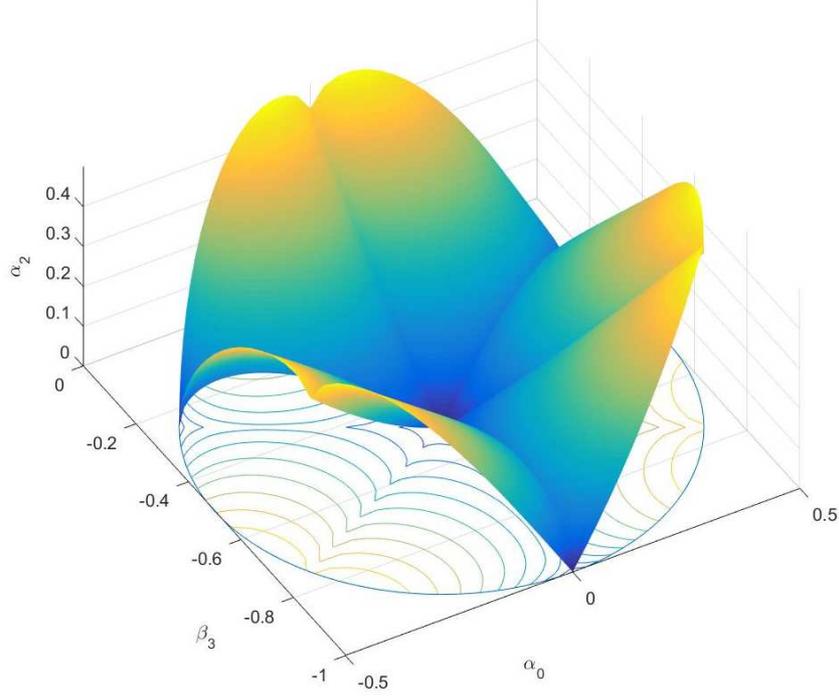}
  \caption{The separatrix below the dome as represented by \eqref{separatrix}.}\label{Fig-sepa}
\end{figure}
Below the dome, the contour plot of the separatrix is illustrated in Figure \ref{Fig-sepa}.   Figure 3 shows a $6$-fold symmetry, which confirms equation (39) of Gaeta and Virga \cite{GaV-16}.
We now contrast  the separatrix and the dome given by \eqref{separatrix} and \eqref{formula-dome} to the same surfaces found numerically in \cite{GaV-16}. To this purpose,
let
\begin{equation}\label{bbbbb}
    \alpha_0 = \rho\cos\chi \qquad\text{ and }\qquad
    \beta_3 = -\tfrac{1}{2} + \rho\sin\chi,
\end{equation}
with $\rho\in[0,\frac{1}{2}]$ and $\chi\in(-\pi,\pi]$.
In Figure \ref{Fig8},
we illustrate the cross-sections for $\chi = -\frac{\pi}{2}, -\frac{5\pi}{12},
-\frac{\pi}{3}, -\frac{\pi}{4}, -\frac{\pi}{6}$ of both the dome and the separatrix, in dash-dot lines and solid lines, respectively.
It can easily be seen that  Figure \ref{Fig8} is consistent with  Figure 8 in \cite{GaV-16}.
\begin{figure}[!btp]
  \centering\includegraphics[width=.6\textwidth]{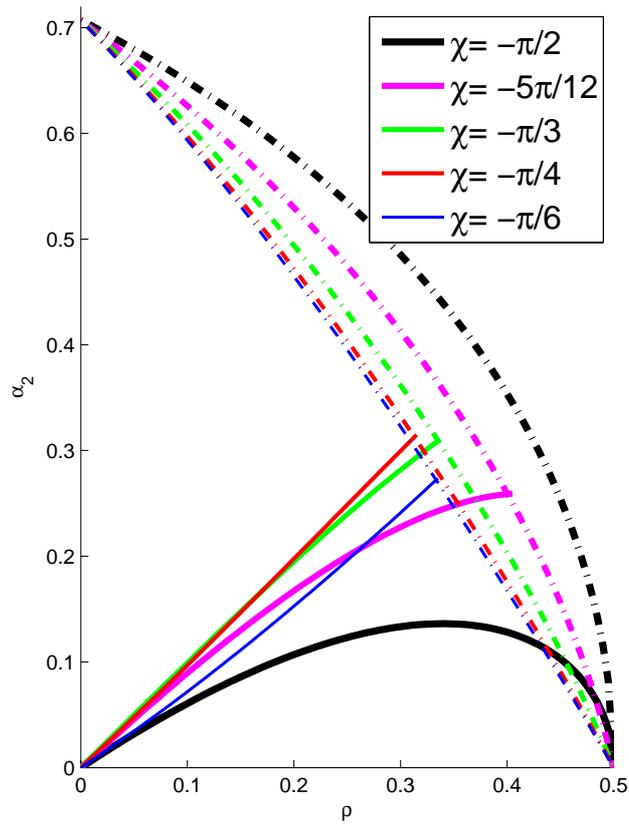}
  \caption{The cross-sections of dome and separatrix corresponding to the cases worked out numerically in \cite{GaV-16}.}\label{Fig8}
\end{figure}

It readily follows from \eqref{bbbbb} that separatix represented by \eqref{separatrix} intersects the plane $\alpha_2=0$ along the circles whose radii are the positive roots of the polynomial
\begin{equation*}
(\rho+1)^3(2\rho+1)^4\rho^{10}(2\rho-1)(\rho-1)^3,
\end{equation*}
the smallest of which identifies the base of the dome. As also shown in \cite{GaV-16}, the octupolar potential possesses a \emph{monkey} saddle when the parameters are chosen on this circle. Apart from the three points shown in Figure~\ref{Fig-sepa} where the separatrix touches this circle, all other points of the latter do not properly belong to the separatrix defined as the locus that separates regions with four and three maxima of the octupolar potential: they will be considered as spurious points, and so discarded in the following.

Next, to reduce \eqref{separatrix} to a simpler form, we study the special case where we set $\chi=-\frac{\pi}{2}$ in \eqref{bbbbb}, so as to describe a cross-section of the separatrix that reaches the base of the dome.
The equation $c_3(\alpha_0,\beta_3,\alpha_2)=0$ for the dome reduces to
\begin{equation*}
    -4(\rho+1)(2\alpha_2^2 -(1-\rho)(1+2\rho)^2)(\alpha_2^2 (\rho-2)-2\rho+1)^2 = 0.
\end{equation*}
Clearly, $\rho+1>0$.
Because $(1-\rho)(1+2\rho)^2$ is monotonically increasing in $\rho\in[0,\frac{1}{2}]$,
we get $(1-\rho)(1+2\rho)^2 \geq 1$. Hence, by $2\alpha_2^2 -(1-\rho)(1+2\rho)^2=0$ and $\alpha_2\geq0$,
we have that the region in parameter space where $\alpha_2 \geq \frac{1}{\sqrt{2}}$ lies above the dome. Moreover,
from $\alpha_2^2 (\rho-2)-2\rho+1=0$, we know that for $\chi=-\frac{\pi}{2}$ the cross-section of the dome is the curve
\begin{equation}\label{dome-I}
    \alpha_2^{(\mathrm{dome})}(\rho) = \sqrt{\frac{1-2\rho}{2-\rho}}.
\end{equation}
For $\chi=-\frac{\pi}{2}$, the separatrix could be rewritten as
\begin{eqnarray*}
   \mathrm{Separatrix} &:& 458752 (1-\rho)^2 (1+\rho)^3
   (\alpha_2^2+4(1-\rho)(1-2\rho)) \left(3(3-\rho)\alpha_2^2-4\rho^2(1-2\rho)\right)^3 \cdot \\
    &&\qquad\qquad{} \left(\alpha_2^4+\alpha_2^2(6\rho+4)+\rho^2(2\rho+1)^2\right) =0.
\end{eqnarray*}
For $\rho\in[0,\frac{1}{2}]$, $(1-\rho)^2 (1+\rho)^3>0$. Moreover,
$\alpha_2^2+4(1-\rho)(1-2\rho)\geq0$ and the equality holds
if, and only if, $(\rho,\alpha_2)=(\frac{1}{2},0)$. Also,
$\alpha_2^4+\alpha_2^2(6\rho+4)+\rho^2(2\rho+1)^2 \geq 0$
and the equality holds if, and only if, $(\rho,\alpha_2)=(0,0)$.
From $3(3-\rho)\alpha_2^2-4\rho^2(1-2\rho)=0$,
we finally obtain that  for $\chi=-\frac{\pi}{2}$ the cross-section of the separatrix is the curve
\begin{equation}\label{sepa-I}
    \alpha_2^{\mathrm{(sepa)}}(\rho) = \frac{2\rho}{\sqrt{3}}\sqrt{\frac{1-2\rho}{3-\rho}}.
\end{equation}
By \eqref{dome-I} and \eqref{sepa-I}, for $\chi=-\frac{\pi}{2}$, the corresponding cross-sections are tangent
at $\rho=\frac{1}{2}$ with vertical tangent.
Since
\begin{equation*}
    [\alpha_2^{\mathrm{(dome)}}(\rho)]^2 - [\alpha_2^{\mathrm{(sepa)}}(\rho)]^2
      = \frac{(3-2\rho)^2(1-\rho-2\rho^2)}{3(2-\rho)(3-\rho)} \geq 0
\end{equation*}
in $\rho\in[0,\frac{1}{2}]$ and the equality holds if, and only if, $\rho=\frac{1}{2}$,  for $\chi=-\frac{\pi}{2}$
the dome  is above the separatrix.

By a similar discussion applied at the case $\chi=-\frac{\pi}{6}$, we also obtain the following curves as representations of the meridian cross-sections of the dome and separatrix, respectively,
\begin{equation}\label{dome-II-and-sepa-II}
    \alpha_2^{(\mathrm{dome})}(\rho) = \frac{1-2\rho}{\sqrt{2}}\sqrt{1+\rho}, \qquad \alpha_2^{\mathrm{(sepa)}}(\rho) = \frac{2\rho}{\sqrt{3}}\sqrt{\frac{1+2\rho}{3+\rho}}.
\end{equation}
They intersect for  $(\rho,\alpha_2)=\left(\frac{1}{3},\frac{\sqrt{2}}{3\sqrt{3}}\right)$.

\section{Conclusion}\label{sec:6}
We studied the octupolar tensor arising from liquid crystal science.
The traceless property of octupolar tensors was shown to be preserved
under orthogonal transformations. The resultant and the E-characteristic
polynomial of the octupolar tensor were constructed explicitly.
Using the resultant theory of algebraic geometry
and the E-characteristic polynomial of the spectral theory of tensors,
we gave an explicit, algebraic expression for the dome and the separatrix, the two significant surfaces for the representation of the octupolar order in three space dimensions. It would be interesting to apply the same algebraic techniques to higher order tensors (or in higher space dimensions) to see whether the pattern of multi-generic states described explicitly in this paper does indeed persist.

\end{document}